\documentclass[11pt,fleqn]{article}
\usepackage{fullpage}
\usepackage{amsmath,amssymb,amsthm}
\usepackage{xspace}
\usepackage[hypertex]{hyperref}

\newcommand{\X}{\mathcal{X}}
\newcommand\ignore[1]{}

\newcommand{\veps}{\varepsilon}
\newcommand{\R}{{\mathbb R}}
\newcommand{\E}{{\mathbb E}}

\newcommand{\ceil}[1]{\lceil{#1}\rceil}

\newcommand{\etal}{{\it{et al. }}}

\newcommand{\iid}{i.i.d.\ }

\DeclareMathOperator{\diam}{diam}

\DeclareMathOperator{\dimC}{dim}
\DeclareMathOperator{\var}{Var}
\DeclareMathOperator{\ddim}{ddim}

\def\Lovasz{Lov{\'a}sz\xspace}

\newtheorem{theorem}{Theorem}[section]
\newtheorem{lemma}[theorem]{Lemma}

\newtheorem{claim}[theorem]{Claim}

\newtheorem{corollary}[theorem]{Corollary}

\theoremstyle{definition}

\title{Dimension reduction techniques for $\ell_p$ ($1 \le p \le 2$), with applications}
\author{
Yair Bartal
\footnote{Hebrew University. 
Work supported in part by an Israel Science Foundation grant \#1609/11.
Email: \texttt{yair@cs.huji.ac.il} }
\and
Lee-Ad Gottlieb
\footnote{
Ariel University.
Email: \texttt{leead@ariel.ac.il} }
}

\begin{document}
\maketitle

\thispagestyle{empty}

\begin{abstract}
For Euclidean space ($\ell_2$), there exists the powerful dimension reduction transform of Johnson and 
Lindenstrauss \cite{JL84}, with a host of known applications. Here, we consider the problem of dimension reduction 
for all $\ell_p$ spaces $1 \le p \le 2$. 
Although strong lower bounds are known for dimension reduction in $\ell_1$, Ostrovsky and 
Rabani \cite{OR02} successfully circumvented these by presenting an $\ell_1$ embedding that maintains 
fidelity in only a bounded distance range, with applications to clustering and nearest neighbor
search. However, their embedding techniques are specific to $\ell_1$ and do not naturally extend
to other norms.

In this paper, we apply a range of advanced techniques and produce bounded range dimension reduction 
embeddings for all of $1 \le p \le 2$, thereby demonstrating that the approach initiated by 
Ostrovsky and Rabani for $\ell_1$ can be extended to a much more general framework.
We also obtain improved bounds in terms of the intrinsic dimensionality.
As a result we achieve improved bounds for proximity problems including 
snowflake embeddings and clustering.
\end{abstract}

\newpage
\setcounter{page}{1}

\section{Introduction}

Dimension reduction for normed space is a fundamental tool for
algorithms and related fields. A much celebrated result
for dimension reduction is the well-known $l_2$ flattening lemma of Johnson
and Lindenstrauss \cite{JL84}: For every $n$-point subset of $l_2$ and every
$0<\veps<1$, there is a mapping into $l_2^k$ that preserves all interpoint
distances in the set within factor $1+\veps$, with target dimension
$k = O(\veps^{-2} \log n)$. The dimension reducing guarantee of
the Johnson-Lindenstrauss (JL) transform is remarkably strong, and has the
potential to make algorithms with a steep dependence on the dimension
tractable. It can be implemented as a simple linear transform, and has proven to
be a very popular tool in practice, even spawning a stream of literature devoted to
its analysis, implementation and extensions (see for example \cite{Al03,AC06,AL11,Ach01,DKS10,VW11}).

Given the utility and impact of the JL transform, it is natural to ask whether
these strong dimension reduction guarantees may hold for other $\ell_p$ spaces as well (see
\cite{In01} for further motivation).
%More generally, given an $n$ point subset of $\ell_p$ what is the minimum $k$ such that this subset embeds with
%constant (or $1+\epsilon$) distortion into $\ell_p^k$, or even into some $k$-dimensional normed space.
This is a fundamental open problem in embeddings, and has attracted significant research. 
Yet the dimension reduction bounds known for $\ell_p$ norms other than
$\ell_2$ are much weaker than those given by the JL transform, and 
it is in fact known that a linear dimension reduction mapping in the style
the JL transform is quite unique to the $\ell_2$-norm
\cite{JN09}. Further, strong lower bounds on dimension reduction are known for
$\ell_1$ \cite{BC05,LN04,ACNN11} and for $\ell_\infty$ \cite{Ma96}, and it
is a plausible conjecture that $\ell_2$ is the only $\ell_p$ space which admits
the strong distortion and dimension properties of the JL transform.

Ostrovsky and Rabani \cite{KOR98,OR00,OR02} successfully circumvented the negative results for $\ell_1$,
and presented a dimension reduction type embedding for the Hamming cube -- and by extension, 
all of $l_1$ -- which preserves fidelity
only in a bounded range of distances. They further demonstrated that their embedding finds use
in algorithms for nearest neighbor search (NNS) and clustering, as these can be reduced to 
subproblems where all relevant distance are found in a bounded range. In fact, 
a number of other proximity problems can also be reduced to bounded range subproblems, including 
the construction of distance oracles and labels \cite{HM06,BGKLR11}, snowflake embeddings
\cite{GK11,BRS11}, and $\ell_p$-difference of a pair of data streams.
Hence, we view a bounded range embedding as a framework for the solution of
multiple important problems. Note also that for spaces with fixed aspect ratio
(a fixed ratio between the largest and smallest distances in the set -- 
a common scenario in the literature, see \cite{IN07}), 
a bounded range mapping is in effect a complete dimension reduction embedding.

The dimension reduction embedding of Ostrovsky and Rabani 
is very specific to $\ell_1$, and does not naturally extend to other norms. 
The central contribution of the paper is to bring advanced techniques to bear on this problem, 
thereby extending this framework to all $1 \le p \le 2$.

\paragraph{Our contribution:} 
We first present a basic embedding in Section \ref{sec:trans}, which shows 
that we can reduce dimension while realizing a certain distance transform
with low distortion.
Using this result, we are able to derive two separate dimension-reducing
embeddings:

\begin{itemize}
\item
Range embedding.
In Theorem \ref{thm:dim-reduction}, we present
an embedding which preserves distances in a given range
with $(1+\veps)$ distortion. The target
dimension is $O(\log n)$, with dependence on the range parameter and $\veps$. This generalizes the 
approach of Ostrovsky and Rabani \cite{OR00,OR02} to all $1 \leq p \leq 2$.
This embedding can be applied in the streaming setting as well, and it can also be modified to achieve
target dimension polynomial in the {\em doubling dimension} of the set (Lemma \ref{lem:fi}).%
\footnote{%
Our range embedding has the additional property that 
it can be used to embed $\ell_p^m$ into $\ell_q^{O(\log n)}$ 
(that is, $m$-dimensional $\ell_p$ into $O(\log n)$-dimensional $\ell_q$, for $1 \le q < p$) with 
$(1+\veps)$-distortion in time $O(m \log n)$, with dependence on the range parameter and 
$\veps$. This is a fast version of the embedding of \cite{JS82}, a common tool for embedding $\ell_p$ into 
more malleable spaces such as $l_1$. 
(The embedding of \cite{JS82} is particularly useful for nearest neighbor search, see \cite{KOR98,HIM12}.)
Note that \cite{JS82} features a large overhead cost $O(m^2)$, and since 
$m$ can be as large as $\Theta(n)$, this overhead can be the most expensive step in algorithms for $\ell_p$.

We also note that the embedding of Theorem \ref{thm:dim-reduction} can be used to produce
efficient algorithms for approximate nearest neighbor search and $\ell_p$ difference, 
although other efficient techniques have already been
developed for these specific problems (see \cite{Kl97,IM98,KOR98,A09,DIIM04,Pa06,N13} for NNS, and 
\cite{FKSV03,FS01,AMS96,AGMS99,IW05,IW05,In06,Li08} for $\ell_p$ difference).}
\item
Snowflake embedding.
An $\alpha$-snowflake embedding is one in which each inter-point distance $t$ in the origin space
is replaced by distance $t^{\alpha}$ in the embedded space (for $0 <\alpha <1$).
It was observed in \cite{GK11,BRS11} that the snowflake of a finite metric space in $l_2$ may be embedded
in dimension which is close to the intrinsic dimension of the space (measured by its doubling dimension), 
and this may
be independent of $n$. In \cite{GK11} the case of $l_1$ was considered as well, however the resulting dimension had
doubly exponential dependence on the doubling dimension.
We demonstrate that the basic embedding can be used to build
a snowflake for $\ell_p$ for all $1 \leq p \leq 2$ with dimension polynomial in the doubling dimension; 
this is found in Lemma
\ref{lem:snowflake}. For $\ell_1$ this provides a {\em doubly exponential} improvement over the 
previously known dimension bound \cite{GK11}, while generalizing the 
$p \in \{1,2\}$ results of \cite{GK11,BRS11} to all $1 \leq p \leq 2$. 
\end{itemize}

\noindent{\em Application to clustering.}
To highlight the utility of our embeddings, we demonstrate 
(in Section \ref{sec:cluster}) their applicability in deriving
better runtime bounds for clustering problems:

We first consider the $k$-center problem, and show that our snowflake embedding can 
be used to provide an efficient algorithm.
For $\ell_p$-spaces of low doubling dimension and fixed $p$, $1 \le p \le 2$, 
we apply our range embedding in conjunction with the 
clustering algorithm of Agarwal and Procopiuc \cite{AP02}, and 
obtain a $(1+\veps)$-approximation to the $k$-center problem in time 
$O(n (2^{\tilde{O}(\ddim(S))} + \log k)) 
+ (k \cdot \veps^{-\ddim(S)/\veps^2})^{k^{1-(\veps/\ddim(S))^{O(1)}}}$. 
%For low doubling dimension, this is subexponential in $k$ and better than the 
%$k^{O(k/\veps^2)} \cdot dn$ runtime achievable via core-sets \cite{BHI02}.

We then consider the min-sum clustering problem, and show that our 
snowflake embedding can be used to provide an efficient algorithm for this problem.
For $\ell_p$-spaces of low doubling dimension, we apply our snowflake 
embedding in conjunction with the clustering algorithm of Schulman 
\cite{Schulman00}, and obtain a $(1+\veps)$-approximation to the 
$k$-center problem in randomized time
$n^{O(1)} + 2^{2^{(O(d'))}} (\veps \log n)^{O(d')}$
where $d' = (\ddim / \veps)^{O(1)}$.

\paragraph{Related work.}
For results on dimension reduction for $\ell_p$ spaces,
see \cite{Sc11} for $p<2$, and also \cite{Sc87,Ta90} for $p=1$,
\cite{Ba90,Ta95} for $1<p<2$, and \cite{Ma96} for $p=\infty$.
Other related notions of dimension reduction have been suggested in the literature.
Indyk \cite{In06} devised an analogue to the JL-Lemma which uses $p$-stable distributions
to produce estimates of interpoint distances; strictly speaking, this is not an embedding into $\ell_p$
(e.g.\ it uses median over the coordinates).
Motivated by the nearest neighbor search problem,
Indyk and Naor \cite{IN07} proposed a weaker form of dimension reduction,
and showed that every doubling subset $S\subset \ell_2$ admits
this type of dimension reduction into $\ell_2$ of dimension
$O(\ddim(S))$.
Roughly speaking, this notion is weaker in that
distances in the target space are allowed
to err in one direction (be too large) for all but one pair of points.
Similarly, dimension reduction into ordinal embeddings (where only relative distance is
approximately preserved) was considered in \cite{ABDFHS08,BDHSZ08}.
Bartal, Recht and Schulman \cite{BRS11} developed a variant of the JL-Lemma
that is local -- it preserves the distance between every point and
the $\hat k$ points closest to it.
Assuming $S\subset \ell_2$ satisfies a certain growth rate condition,
they achieve, for any desired $\hat k$ and $\veps>0$, an embedding of this type
with distortion $1+\veps$ and dimension $O(\veps^{-2}\log \hat k)$.

\section{Preliminaries}\label{sec:preliminaries}

\paragraph{Embeddings and metric transforms.}
Following \cite{BES06}, we define an {\em oblivious} embedding to be an embedding
which can be computed for any point of a database set $X$ or query set $Y$,
without knowledge of any other point in $X$ or $Y$. (This differs slightly from
the definition put forth by Indyk and Naor \cite{IN07}.) Familiar
oblivious embeddings include standard implementations of the JL-Lemma for $l_2$, the dimension reduction mapping of Ostrovsky
and Rabani \cite{OR02} for the Hamming cube, and the embedding of Johnson and Schechtman \cite{JS82} for $\ell_p$,
$p \le 2$.
A {\em transform} is a function mapping from the positive reals to the positive reals, and
a \emph{metric transform} maps a metric distance function to another metric distance function on the same set
of points. 
%We say that a metric transform is $b$-\emph{bounded} ($b>0$) if it always
%results in a distance function wherein all interpoint distances are bounded
%by $b$. Given an $b$-bounded metric
%transform $H$, for every $s>0$ there is an $s$-bounded metric transform $H^{(s)}(t) = (s/b)\cdot H(b/s \cdot t)$ 
%with the same smoothness. 
An embedding is {\em transform preserving} with respect to a transform if it achieves the distances defined by 
that transform.

\paragraph{Range Embedding.}
Let $(X,d_X),(Y,d_Y)$ be metric spaces. For distance scales $0 \leq a \leq b \leq \infty$, an $[a,b]$-embedding of $X$
into $Y$ with distortion $D$ is a mapping $f:X\rightarrow Y$ such that for all
$x,y \in X$ such that $d_X(x,y) \in [a,b]$:
$ 1 \leq c \cdot \frac{d_Y(f(x),f(y))}{d_X(x,y)}   \leq D.$
(Here, $c$ is any scaling constant.)
Then $f$ is a \emph{range embedding} with range $[a,b]$.
If $f$ has the additional property that for all $x,y$ such that $d(x,y) < a$ we have
$c \cdot \frac{d_Y(f(x),f(y))}{a} \leq D$,
then we say that $f$ is range preserving from below.
Similarly, if for all $x,y$ such that $d(x,y) > b$ we have
$c \cdot \frac{d_Y(f(x),f(y))}{b} \geq 1$,
then we say that $f$ is range preserving from above.
And if $f$ is range preserving from above and below, then we say that $f$ is a {\em strong} range embedding.

Let $R> 1$ be a parameter. We say that $X$ admits an $R$-range embedding into $Y$ with distortion 
$D$ if for every $u>0$ there exists an $[u,u R]$ embedding of $X$ into $Y$ with distortion $D$. As 
above, an $R$-range embedding may be range preserving from above or below. We will usually take
$u=1$.

\paragraph{Nets and hierarchies.}
Given a metric space $S$, $S' \subset S$ is a $\gamma$-net of $S$ if the minimum interpoint distance in $S'$ is at least
$\gamma$, while the distance from every point of $S$ to its nearest neighbor in $S'$ is less than $\gamma$.
Let $S$ have minimum inter-point distance 1.
A {\em hierarchy} is a series of $\lceil \log \Delta \rceil$ nets ($\Delta$ being the aspect ratio of $S$), 
where each net $S_i$ is a $2^i$-net of the previous net $S_{i-1}$.
The first (or bottom) net is $S_0 = S$,
and the last (or top) net $S_t$ contains a single point called the {\em root}.
For two points $u \in S_i$ and $v \in S_{i-1}$, if $d(u,v)<2^i$ then we say that $u$ {\em covers} $v$, and 
this definition allows $v$ to have multiple covering points in $S_i$.
The closest covering point of $v$ is its {\em parent}.
The distance from a point in $S_i$ to its ancestor in $S_j$ is at most
$\sum_{k=i+1}^j 2^k = 2 \cdot (2^j - 2^{i+1}) < 2 \cdot 2^j$.

Given $S$, a hierarchy for $S$ can be built in time 
$2^{O(\ddim(S))}n$, and this term bounds the size of the hierarchy as well \cite{HM06,CG06}.
The height of the hierarchy is 
$O (\min \{n, \log \Delta \})$.

\paragraph{Doubling dimension.}
For a metric $(\X,\rho)$, let $\lambda>0$
be the smallest value such that every
ball in $\X$ can be covered by $\lambda$ balls of half the radius.
The {\em doubling dimension} of $\X$ is $\ddim(\X)=\log_2\lambda$.
Note that $\ddim(\X) \le \log n$.
The following packing property can be shown (see, for example \cite{KL04}):
Suppose that $S\subset\X$ has a minimum interpoint distance
of at least $\alpha$. Then
$ |S| \leq \Big(\tfrac{2\diam(S)}{\alpha}\Big)^{\ddim(\X)}. $

\paragraph{Probabilistic partitions.}
Probabilistic partitions are a common tool used in embeddings. Let $(X,d)$
be a finite metric space. A partition 
$P$ of $X$ is a collection of
non-empty pairwise disjoint clusters $P=\{C_1,C_2,\ldots,C_t\}$ such
that $X=\cup_jC_j$. For $x\in X$ we denote by $P(x)$ the cluster
containing $x$.
We will need the following decomposition lemma due to Gupta, Krauthgamer
and Lee \cite{GKL03}, Abraham, Bartal and Neiman \cite{ABN08}, and Chan,
Gupta and Talwar \cite{CGT10}.\footnote{%
\cite{GKL03} provided slightly different quantitative bounds
than in Theorem \ref{thm:decomp}.
The two enumerated properties follow, for example, from Lemma 2.7 in
\cite{ABN08}, and the bound on support-size $m$ follows by an application
of the \Lovasz\ Local Lemma sketched therein.
}
Let ball
$B(x,r) = \{y| \, \|x-y\| \le r \}$.

\begin{theorem}[Padded Decomposition of doubling metrics \cite{GKL03,ABN08,CGT10}]
\label{thm:decomp}
There exists a constant $c_0>1$,
such that for every metric space $(X,d)$, every $\veps\in(0,1)$,
and every $\delta > 0$, there is a
multi-set $\mathcal D = [P_1, \ldots, P_m]$ of partitions of $X$,
with $m \leq c_0 \veps^{-1} \dim(X)\log \dim(X)$, such that
\begin{enumerate} 
\item Bounded radius: $\diam(C) \leq \delta$ for all clusters $C \in \bigcup_{i=1}^m P_i$.
\item Ball padding: If $P$ is chosen uniformly {from} $\mathcal D$, then for all $x \in X$,
$$
\Pr_{P\in\mathcal D} [B(x,\tfrac{\delta}{c_0\dim(X)}) \subseteq P(x)] \geq 1-\veps.
$$
\end{enumerate}
\end{theorem}

\paragraph{Stable distributions.}
The density of a symmetric $p$-stable random variable ($0 < p \le 2$) is
$h(x) = \frac{1}{\pi} \int_0^\infty \cos(tx)e^{-t^p} dt$ \cite[Section 2.2]{Zo86}.
The density function $h(x)$ is unimodal \cite{SY78,Ya78,Ha84} and bell-shaped \cite{Ga84}.
It is well known that
$\frac{c_{p}}{1+x^{p+1}} \le h(x) \le \frac{c_{p}'}{1+x^{p+1}}$ for constants
$c_{p},c_{p}'$ that depend only on $p$ \cite{No12} (and we may ignore the dependence
on $p$ for the purposes of this paper).
Using this approximation for $h(x)$, it is easy to see that for $0 < q < p$ and $p$-stable
random variable $g$,
%\begin{equation}
$\mathbb{E} [g^q]
= \int_0^\infty x^q h(x) dx
\approx \int_0^1 x^q dx + \int_1^\infty x^{q-(p+1)} dx
\approx \frac{1}{p-q}.$
(Here we use the notation $\approx$ in the same sense as $\Theta(\cdot)$.)
%\end{equation}
Also, for $b>1$,
%\begin{equation}
$\int_0^b x^p h(x) dx
\approx \int_0^1 x^p dx + \int_1^b \frac{1}{x} dx
\approx 1 + \ln b.$
%\end{equation}
Let $g_1,g_2,\ldots,g_m \in G$ be a set of \iid symmetric $p$-stable random variables, and let $v$
be a real $m$-length vector. A central property of $p$-stable random variables is that
$\sum_{j=1}^m g_jv_j$
is distributed as
$g (\sum_{j=1}^m  v_j^p)^{1/p} = g \|v\|_p$, for all $g \in G$.
When all $g_i \in G$ are normalized as $\mathbb{E}[|g_i|^q]=1$, we have that
%\begin{equation}
$\mathbb{E} \left[ \left| \sum_{j=1}^m g_jv_j \right|^q \right]
= \mathbb{E} [|g|^q] \|v\|_p^q
= \|v\|_p^q.$
%\end{equation}

\section{Basic transform-preserving embedding}\label{sec:trans}

%Dimension reduction for $\ell_p$, $1 \le p \le 2$
%In this section, we present an oblivious strong range embedding for $\ell_p$, $1 \le p \le 2$.

In Theorem \ref{thm:full} below, we present an embedding which realizes a certain distance
transform with low distortion,
while also reducing dimension to $O(\log n)$ (with dependence on the range and the desired distortion). 
We then demonstrate that the transform itself has several desirable properties.
In the next section, we will use this transform-preserving embedding 
to obtain a range embedding and a snowflake embedding.

%and as a result we 
%obtain a dimension reducing embedding 
%(Theorem \ref{thm:dim-reduction}).
%To prove Theorem \ref{thm:full}, 

\subsection{Embedding}
We first present a randomized embedding into a single coordinate, and show that it is
transform-preserving in expectation only (Lemma \ref{lem:H}). 
We then show that a concatenation of many single-coordinate embeddings yields a
single embedding which (approximately) preserves the bounded transform with 
high probability, and this gives Theorem \ref{thm:full}.

The following single-coordinate embedding is inspired by the Nash device of Bartal \etal \cite{BRS11} 
(see also \cite{RR07}), and is related to the spherical threshold function of Mendel and Naor \cite[Lemma 
5.9]{MN04}. Broadly speaking, 
our embedding uses the sine function as a dampening tool, which serves to mitigate undesirable 
properties of $p$-stable distributions (i.e., their heavy tails). Our central contribution 
is to give a tight analysis of the transform preserved by our embedding.\footnote{%
It may be possible to replace our embedding with that of Mendel and Naor \cite{MN04}, but one would still 
require a tighter analysis, and the final dimension would likely increase.
Note also that our embedding is into the reals, while \cite{MN04} 
embed into the complex numbers; as $\ell_p^d$ over $\mathbb{C}$ embeds into 
$\ell_p^{O(\epsilon^{-2}\sqrt{p}2^{p/2}d)}$ over $\mathbb{R}$ with distortion $1+\epsilon$ (a 
consequence of Dvoretzky's theorem \cite{MS86}), the embedding of \cite{MN04} can in fact be 
used to achieve an embedding into the reals with increased dimension.
}
%This allows us to 
%show that like \cite{JS82}, the embedding can be used to embed $\ell_p$ into $l_q$ (for $1 \le q \le 
%p \le 2$), with arbitrarily small distortion, while simultaneously reducing dimension as well.

Let $S \subset \ell_p$ be a set of $m$-dimensional vectors for $1\le p \le 2$. Let
$g_j \in G$ be $k$ i.i.d.\ symmetric $p$-stable random variables, and fix parameters $s$ (the
{\em threshold}) and $0 \le \phi \le 2\pi$. Further define the fixed constant
$P_q = \mathbb{E}[|\cos \theta|^q] = \frac{1}{2\pi} \int_0^{2\pi} |\cos \theta|^q d\theta$
for $1 \le q \le p$ (and note that
$\frac{1}{2}
= \frac{1}{2\pi} \int_0^{2\pi} \cos^2 \theta d\theta
\le P_q
\le \frac{1}{2\pi} \int_0^{2\pi} |\cos \theta| d\theta
= \frac{2}{\pi}$).
Then the embedding $F_{\phi,s}: S \to L_q^1$ (i.e., into a single coordinate of $L_q$) for
vector $v \in S$ is defined by

\begin{equation}
F_s(v) = F_{\phi,s}(v) = \frac{s}{2P_q^{1/q}} \sin \left( \phi+\frac{2}{s}\sum_{i=1}^m g_iv_i \right).
\end{equation}

Note that $0 \le F_s(v) < s$. For vectors $v,w \in S$ we have that

\begin{eqnarray*}
|F_s(v)-F_s(w)|^q
&=& \frac{s^q}{2^qP_q} \left|
\sin \left( \phi+\frac{2}{s}\sum_{i=1}^m g_iv_i \right)
- \sin \left( \phi+\frac{2}{s}\sum_{i=1}^m g_iw_i \right)
\right|^q	\\
&=& \frac{s^q}{P_q} \left|
\sin \left( \frac{1}{s} \sum_{i=1}^m g_i(v_i-w_i) \right)
\cos \left( \phi + \frac{1}{s} \sum_{i=1}^m g_i(v_i+w_i) \right)
\right|^q.
\end{eqnarray*}

Now, when $\phi$ is a random variable chosen uniformly from the range $[0,2\pi]$
and independently of set $G$, we have that

\begin{eqnarray*}
\mathbb{E}[|F_s(v)-F_s(w)|^q]
&=& \frac{s^q}{P_q} \mathbb{E} \left[ \left|
	\sin \left( \frac{1}{s} \sum_{i=1}^m g_i(v_i-w_i) \right)
	\cos \left( \phi + \frac{1}{s} \sum_{i=1}^m g_i(v_i+w_i) \right)
	\right|^q \right]	\\
&=& \frac{s^q}{P_q} \mathbb{E} \left[ \left|
	\sin \left( \frac{1}{s} \sum_{i=1}^m g_i(v_i-w_i) \right)
	\cos \left( \phi \right)
	\right|^q \right]	\\
&=& \frac{s^q}{P_q} \mathbb{E} 
	\left[ \left|
	\sin \left( \frac{1}{s} \sum_{i=1}^m g_i(v_i-w_i) \right) 
	\right|^q \right]
	\cdot
	\mathbb{E} \left[ \left| 
	\cos \left( \phi \right) 
	\right|^q \right]	\\
&=& s^q \mathbb{E} \left[ \left|
	\sin \left( \frac{1}{s} \sum_{i=1}^m g_i(v_i-w_i) \right)
	\right|^q \right]
\end{eqnarray*}
where the second equality follows from the periodicity of the cosine function, and
the independence of $\phi$ and $G$.

This is our single-coordinate embedding. In Lemma \ref{lem:H} below, we will describe its behavior on
interpoint distances -- that is, we derive useful bounds on 
$\mathbb{E}[|F_s(v)-F_s(w)|^q$ . 
Now $\frac{1}{s} \sum_{i=1}^m g_i(v_i-w_i)$ is distributed as $g \frac{\|v-w\|_p}{s}$ 
for random variable $g \in G$, so we will set $a = \frac{\|v-w\|_p}{s}$ and will
derive bounds for 
$H(a) = \mathbb{E}[|\sin(ag)|^q] = s^{-q} \mathbb{E}[|F_s(v)-F_s(w)|^q]$.
But first we introduce the full embedding $f$, which
is a scaled concatenation of $k$ single-coordinate embeddings: Independently for each coordinate
$i$, create a function $F_{\phi_i} = F_{\phi_i,s}$ by fixing a random angle $0 \le \phi_i \le 2\pi$ and a
family of $k$ i.i.d.\ symmetric $p$-stables. Then coordinate $f(v)_i$ is defined by
$$f(v)_i = k^{-1/p} F_{\phi_i,s}(v)$$ which can be constructed in $O(m)$ time per coordinate.
Note that for $t = \| v-w \|_p$,
$\E[|f(v)_i-f(w)_i|^q]
= \frac{1}{k} \E[|F_{\phi_i}(v) - F_{\phi_i}(w)|^q]
= \frac{s^q}{k} H(t/s)$, and so
$$
\E[\| f(v)-f(w) \|_p^q]
= s^q H(t/s).
$$
Below, we will show that with high probability $f$ is 
is transform-preserving with respect to $H$
with low distortion.

%(Note that the $q$-th power will be of interest 
%for the full embedding, when we concatenate many single-coordinate embeddings into
%a single vector, and then consider the $\ell_q$ distance between $v,w$.)

\subsection{Analysis of basic embedding}

We now show that both the single-coordinate embedding and the full embedding
have desirable properties.
Recall that $h(u)$ is the density function of $p$-stable random variables. Set 
$$Q = 2\int_0^\infty u^q h(u) du \approx \frac{1}{p-q}.$$ 
Also, for $a<1$ and some fixed $a^2<\veps<1$, set
$$Q_a
= \frac{1}{2}\int_0^{\sqrt{\veps}/a} u^p h(u) du
= \frac{1}{2} \left[ \int_0^1 u^p h(u) du + \Theta(\ln(\sqrt{\veps}/a)) \right]
\approx 1 + \ln(\sqrt{\veps}/a).$$

\begin{lemma}\label{lem:H}
Let $g$ be a symmetric $p$-stable random variable.
For $1 \le q \le p \le 2$ and any fixed $0 < \veps < \frac{1}{2}$,
$H(a)= \mathbb{E}[|\sin(ag)|^q]$ obeys the following:
\begin{enumerate}
\renewcommand{\theenumi}{(\alph{enumi})}
\item\label{it:lem-threshold}
Threshold: $H(a) \le 1$.
\item\label{it:lem-bilip}
Bi-Lipschitz for small scales: When $q<p$ and
$a \le \min \{
\veps^{\frac{1}{2}+\frac{1}{p-q}},
\sqrt{\veps}(1+(p-q)\veps^{-(q/2+1)})^{-1/(p-q)}
\}$,
we have
$1-O(\veps) \le \frac{H(a)}{a^qQ} \le 1+O(\veps).$
When $q=p$ and
$a \le \sqrt{\veps} e^{-\veps^{-(\frac{q}{2}+1)}}$,
we have
$1-O(\veps) \le \frac{H(a)}{a^q Q_a} \le 1+O(\veps).$
\item\label{it:lem-bilip2}
Bi-Lipschitz for intermediate scales: When $q<p$ 
and $a<1$, let $\delta=1-a^{p-q}$), and we have
$H(a) = \Theta(1+\delta Q)a^q.$
When $q=p$ we have
$H(a) = \Theta \left(1+\ln(1/a) \right) a^q.$
\item\label{it:lem-large}
Lower bound for large scales: When $a \ge 1$, $H(a) > \frac{1}{8}$.
%\item\label{it:lem-mono}
%Monotonicity: $H(a)$ is monotone increasing in $a$.
\item\label{it:lem-smooth}
Smoothness: When $a \le 1$,
$\frac{|H((1+\veps)a)-H(a)|}{H(a)} = O(\veps).$
\end{enumerate}
\end{lemma}

It follows that the distance transform implied by our single-coordinate
embedding $F_{s}$ achieves the bounds of Lemma \ref{lem:H} scaled by $s^q$,
if only in expectation. 
Item \ref{it:lem-bilip} implies that for very small $a$ (i.e., when the inter-point
distance under consideration is sufficiently small with respect to the
parameter $s$) the embedding has very small expected distortion 
(at least when $q<p$).
Weaker distortion bounds hold for distances smaller than $s$
(item \ref{it:lem-bilip2}).
For distances greater than $s$, these are contracted to 
the threshold (item \ref{it:lem-threshold}) or
slightly smaller (item \ref{it:lem-large}). 
The smoothness property will be useful for
constructing the snowflake in Section \ref{sec:snowflake}.
We proceed to consider the full embedding:

\begin{theorem}\label{thm:full}
Let $1 \le q \le p \le 2$ and $0 < \veps < \frac{1}{2}$,
and consider an $n$-point set $S \in l^m_p$. Set threshold $s>1$.
Then with constant probability the oblivious embedding $f: S \to l_q^k$ for
$k= O \left( \frac{\log n}{\veps^2} \cdot \min \left\{
s^{2q}, \max \left\{ \frac{s^{2q-p}}{2q-p},\veps s^q \right\}
\right\} \right)$,
satisfies the following for each point pair $v,w \in S$, where $t=\| v-w \|_p$:
\begin{enumerate}
\renewcommand{\theenumi}{(\alph{enumi})}
\item \label{it:thm-threshold}
Threshold:
$\| f(v)-f(w) \|_q^q < s^q.$
\item \label{it:thm-bilip}
Bi-Lipschitz for large scales:
When $t \ge 1$, we have
$$(1-\veps) s^qH(t/s)
\le \|f(v)-f(w)\|_q^q
\le (1+\veps) s^qH(t/s).$$
\item \label{it:thm-expansion}
Bounded expansion for small scales: When
$t < 1$, we have
$$\|f(v)-f(w)\|_q^q \le s^qH(1/s) + \veps.$$
\end{enumerate}
The embedding can be constructed in $O(mk)$ time per point.
\end{theorem}

Theorem \ref{thm:full} demonstrates that in a certain range, there
exists a transform-preserving embedding with high probability.
(We note that when $2q$ is close to $p$, better dimension bounds can be obtained 
by embedding into an interim value $q+c$ for some $c$ and then embedding into $q$.)

\begin{proof}[Proof of Lemma \ref{lem:H}]
{\bf Item \ref{it:lem-threshold}.}
This follows trivially from the fact that $|\sin(x)| \le 1$.

\noindent {\bf Item \ref{it:lem-bilip}.}
Note that since the density function $h$ is symmetric about 0,
\begin{equation}\label{eq:terms}
H(a)
= 2 \int_0^{\sqrt{\veps}/a} |\sin(au)|^{q} h(u) du
+ 2 \int_{\sqrt{\veps}/a}^\infty |\sin(au)|^{q} h(u) du.
\end{equation}

We show that under the conditions of the item, the second term in Equation (\ref{eq:terms}) is
dominated by the first. Considering the second term, we have that
\begin{equation}
2 \int_{\sqrt{\veps}/a}^\infty |\sin(au)|^{q} h(u) du
< 2\int_{\sqrt{\veps}/a}^\infty h(u) du
< 2c_p' \int_{\sqrt{\veps}/a}^\infty \frac{1}{u^{p+1}} du
= \frac{2c_p'}{p} \left( \frac{a}{\sqrt{\veps}} \right)^{p}.
\end{equation}

Turning to the first term, recall the Taylor series expansion
$\sin(x) = x - \frac{x^3}{3!} + \frac{x^5}{5!} - \ldots$;
so when $x < \sqrt{\veps}$ we have that $x(1-\frac{\veps}{6}) < \sin(x) < x$.
Also note that the conditions of the item give that $a < \sqrt{\veps}$,
and so $\frac{\sqrt{\veps}}{a} >1$. Hence, when $q<p$ we have

\begin{eqnarray*}
2 \int_0^{\sqrt{\veps}/a} |\sin(au)|^{q} h(u) du
&>&  2 (1-\veps/6)^{q} a^q \int_0^{\sqrt{\veps}/a} u^q h(u) du \\
&>&  2 (1-\veps/3) a^q c_p \left[
      \int_0^1 \frac{u^q}{1+u^{p+1}} du
    + \int_1^{\sqrt{\veps}/a} \frac{u^q}{1+u^{p+1}} du	
      \right]	\\
&>&  2 (1-\veps/3) a^q c_p \left[
      \int_0^1 \frac{u^q}{2} du
    + \int_1^{\sqrt{\veps}/a} \frac{u^{q-p-1}}{2} du	
      \right]	\\
&=&  (1-\veps/3) a^q c_p \left[
      \frac{1}{q+1} + \frac{1-(a/\sqrt{\veps})^{p-q}}{p-q}
      \right]
\end{eqnarray*}

When $q=p$, the same analysis gives

\begin{eqnarray*}
2 \int_0^{\sqrt{\veps}/a} |\sin(au)|^{q} h(u) du
&>&  (1-\veps/3) a^q c_p \left[
      \frac{1}{q+1} + \ln(\sqrt{\veps}/a)
      \right]
\end{eqnarray*}

We proceed to show that the first term of Equation (\ref{eq:terms}) exceeds the second by a factor of
$\Omega(\veps^{-1})$. For $q=p$ this holds trivially when
$a^q \ln(\sqrt{\veps}/a) \ge \frac{1}{\veps} \left( \frac{a}{\sqrt{\veps}} \right)^p$ --
or equivalently, when $\ln(\sqrt{\veps}/a) \ge \veps^{-(\frac{q}{2}+1)}$ --
which in turn holds exactly when $a \le \sqrt{\veps} e^{-\veps^{-(\frac{q}{2}+1)}}$.
For $q<p$, the condition is fulfilled when
$a^q \ge \frac{1}{\veps} \left( \frac{a}{\sqrt{\veps}} \right)^p$, which holds exactly when
$a \le \veps^{\frac{p/2+1}{p-q}}$.
Better, the condition is also fulfilled when
$a^q \left[ \frac{1-(a/\sqrt{\veps})^{p-q}}{p-q} \right] \ge \frac{1}{\veps} \left( \frac{a}{\sqrt{\veps}} \right)^p$,
and this is equivalent to satisfying
$a^{p-q} \left[ \frac{p-q}{\veps^{\frac{p}{2}+1}} + \frac{1}{\veps^{(p-q)/2}} \right]
\le 1$; this holds exactly when
$a \le \sqrt{\veps}(1+(p-q)\veps^{-(q/2+1)})^{-1/(p-q)}$.

As the first term of Equation (\ref{eq:terms}) dominates the second, 
it follows that when $q<p$ then for some constant $c$
\begin{eqnarray*}
H(a)
&=& 2 \int_0^{\sqrt{\veps}/a} |\sin(au)|^{q} h(u) du
	+ 2 \int_{\sqrt{\veps}/a}^\infty |\sin(au)|^{q} h(u) du	\\
&\le& 2(1+c\veps) \int_0^{\sqrt{\veps}/a} |\sin(au)|^{q} h(u) du	\\
&<&   2(1+c\veps) a^q \int_0^{\infty} u^q h(u) du 	\\
&=&   (1+c\veps) a^q Q.
\end{eqnarray*}

Further, we have that
\begin{eqnarray*}
H(a)
&=& 2 \int_0^{\sqrt{\veps}/a} |\sin(au)|^{q} h(u) du
	+ 2 \int_{\sqrt{\veps}/a}^\infty |\sin(au)|^{q} h(u) du	\\
&>& 2\int_0^{\sqrt{\veps}/a} |\sin(au)|^{q} h(u) du    \\
&>& 2(1-\veps/3) a^q \int_0^{\sqrt{\veps}/a} u^q h(u) du	\\
&>& (1-\veps/3)(1-c' \veps) a^q Q,
\end{eqnarray*}
where the final inequality follows from noting that since
$a \le \veps^{\frac{1}{2}+\frac{1}{p-q}}$,
$\int_{\sqrt{\veps}/a}^\infty u^q h(u) du
\le \int_{\veps^{-1/(p-q)}}^\infty u^q h(u) du
\approx \frac{\veps}{p-q}
\approx \veps Q$,
and so
$\int_0^{\sqrt{\veps}/a} u^q h(u) du
= \int_0^\infty u^q h(u) du - \int_{\sqrt{\veps}/a}^\infty u^q h(u) du
> (1-c'\veps)Q$
for some $c'$.

This completes the analysis for $q<p$. The same analysis gives that when $q=p$,
$$H(a) < 2(1+c\veps) a^q \int_0^{\sqrt{\veps}/a} u^q h(u) du
= (1+c\veps) a^q Q_a,$$
and
$$H(a) > 2(1-\veps/3) a^q \int_0^{\sqrt{\veps}/a} u^q h(u) du
= (1-\veps/3) a^q Q_a.$$

\noindent{\bf Item \ref{it:lem-bilip2}.} The analysis is similar to that presented in the proof of Item
\ref{it:lem-bilip}. Noting that when $0 \le x \le 1$, $\sin(x) \approx x$, and recalling that under the conditions of
the item $a \le 1$, we have for $q<p$ that

\begin{eqnarray*}
H(a)
&=&	2 \int_0^{1/a} |\sin(au)|^{q} h(u) du
	+ 2 \int_{1/a}^\infty |\sin(au)|^{q} h(u) du. \\
&=& 	\Theta \left( a^q \int_0^{1/a} u^{q} h(u) du \right)
	+ O \left( \int_{1/a}^\infty h(u) du \right)   \\
&=&     \Theta \left( \left( 1+\frac{1-a^{p-q}}{p-q} \right) a^q \right)
	+ O(a^p) \\
&=&     \Theta \left( \left( 1+\frac{\delta}{p-q} \right) a^q \right)
	+ O(a^p) \\
&=& 	\Theta \left( \left( 1+\frac{\delta}{p-q} \right) a^q \right).
\end{eqnarray*}

Similarly, for $q=p$ we have

\begin{eqnarray*}
H(a) =
\Theta \left( a^q \int_0^{1/a} u^{p} h(u) du \right)
+ O \left( \int_{1/a}^\infty h(u) du \right)
= \Theta \left( a^q \int_0^{1/a} u^{p} h(u) du \right)
\approx \left( 1+\ln(1/a) \right) a^p.
\end{eqnarray*}

\noindent{\bf Item \ref{it:lem-large}.}
First note that when $p \ge 1$,
$h(x)
= \frac{1}{\pi} \int_{0}^\infty \cos(tx) e^{-t^p} dt
< \frac{1}{\pi} \int_{0}^\infty e^{-t^p} dt
< \frac{1}{\pi}$, so
$\int_0^b h(u) du < \frac{b}{\pi}$.
Since $h(x)$ is a symmetric density function, we have
$\int_0^\infty h(u) du = \frac{1}{2}$,
and so
$\int_b^\infty h(u) du > \frac{1}{2} - \frac{b}{\pi}$.
We have for any $0 < \theta < \frac{\pi}{2}$,
\begin{eqnarray*}
H(a)
&\ge&	2 \int_{\theta/a}^{\infty} |\sin(au)|^{q} h(u) du	\\
&>&	2 \sum_{i=0}^\infty \int_{\frac{i \pi + \theta}{a}}^{\frac{(i+1)\pi - \theta}{a}} |\sin(au)|^{q} h(u) du	\\
&>&	2 |\sin(\theta)|^q \sum_{i=0}^\infty \int_{\frac{i \pi + \theta}{a}}^{\frac{(i+1)\pi - \theta}{a}} h(u) du	\\
&>&	2 |\sin(\theta)|^q \left[1 - \frac{2\theta}{\pi} \right] 
	\sum_{i=0}^\infty \int_{\frac{i \pi + \theta}{a}}^{\frac{(i+1)\pi + \theta}{a}} h(u) du	\\
&=&	2 |\sin(\theta)|^q \left[1 - \frac{2\theta}{\pi} \right] \int_{\theta/a}^\infty h(u) du	\\
&>&	|\sin(\theta)|^q \left[1 - \frac{2\theta}{\pi} \right] \left[ 1 - \frac{2\theta}{a\pi} \right].
\end{eqnarray*}
Where the fourth inequality follows from the fact that
$h(x)$ is monotone decreasing for $x \ge 0$.
The claimed result follows by taking $a=1$ (the maximum value of $a$) 
and $\theta = \frac{\pi}{4}$, and recalling that $q \le 2$.

\noindent{\bf Item \ref{it:lem-smooth}.}
As in the proof of Item \ref{it:lem-bilip2} above, we have that for $q \le p$
and $a \le 1$,
$H(a)
\approx  \int_0^{1/a} |\sin(au)|^{q} h(u) du + \int_{1/a}^\infty h(u) du
=  \int_0^{1/a} |\sin(au)|^{q} h(u) du + O(a^p)
\approx  \int_0^{1/a} |\sin(au)|^{q} h(u) du$.
Now recall the Taylor series expansion
$\cos(x) = 1 - \frac{x^2}{2!} + \frac{x^4}{4!} - \ldots > 1-\frac{x^2}{2}$ (when $0 \le x \le 1$),
and note that as a consequence of the Mean Value Theorem,
$\left| |A|^q - |B|^q \right| = qC^{q-1}||A|-|B||$ for some $|A| \ge C \ge |B|$.
Further noting that when $u \le \frac{1}{a}$ we have $au \le 1$
and that when $u \le \frac{1}{\veps a}$ we have $\veps au \le 1$,
we conclude that
\begin{eqnarray*}
H(a(1+\veps)) - H(a)
&=& 	2\int_0^{\infty} [|\sin(a(1+\veps)u)|^q - |\sin(au)|^q] h(u) du							\\
&\le &	2\int_0^{\infty} q [\max \{ |\sin(a(1+\veps)u)|,|\sin(au)| \}]^{q-1} \cdot					\\
& &	[||\sin(a(1+\veps)u)| - |\sin(au)||] h(u) du	 								\\
&\le&	2\int_0^{\infty} q [\max \{ |\sin(au)\cos(\veps au)|+|\sin(\veps au)\cos(au)|,|\sin(au)| \}]^{q-1} \cdot	\\
& &	[||\sin(au)\cos(\veps au)|+|\sin(\veps au)\cos(au)| - |\sin(au)||] h(u) du					\\
&\le& 	2\int_0^{\infty} q [|\sin(au)|^{q-1}+|\sin(\veps au)|^{q-1}][|\sin(\veps au)|] h(u) du				\\
&=&	O \left( 	\int_0^{1/a} \veps |\sin(au)|^q h(u) du
	+ \veps a \int_{1/a}^{1/(\veps a)} u h(u) du	
	+ \int_{1/(\veps a)}^\infty h(u) du 	\right)									\\
&=&	O \left( 	\int_0^{1/a} \veps |\sin(au)|^q h(u) du 	
	+ \veps a^p
	+ \veps^p a^p \right)												\\
&=&	O(\veps H(a))
\end{eqnarray*}

\end{proof}

\begin{proof}[Proof of Theorem \ref{thm:full}]
The first claim of the theorem follows from the fact that for all $i$, $0 \le F_{\phi_i,s}(v) < s$.

To prove the rest of the theorem, we may make use of Hoeffding's inequality. When $1 \le t \le s$,
we have by Lemma \ref{lem:H}\ref{it:lem-bilip}\ref{it:lem-bilip2}\ref{it:lem-large} that $s^qH(t/s) = s^q \Omega(1) =
\Omega(1)$. Then Claim \ref{clm:hoeffding} implies that for some
$k = O \left( \frac{s^{2q}}{\veps^2} \log n \right)$,
we have
\begin{eqnarray*}
\Pr \left[ | \| f(v)-f(w)\|_q^q - s^qH(t/s)| > \veps s^qH(t/s) \right]
&=& \Pr \left[ | \| f(v)-f(w)\|_q^q - s^qH(t/s)| > \veps \Omega(1) \right] 	\\
&\le& \frac{1}{n^2},
\end{eqnarray*}
so this distortion guarantee can hold simultaneously for all point pairs.
When $t < 1$, first note that when $H(t/s) = \Theta(1)$, we can use the same proof as for
$1 \le t \le s$ above. If $H(t/s) = o(1)$, we have
\begin{eqnarray*}
\Pr \left[ \| f(v)-f(w)\|_q^q > s^qH(1/s) + \veps \right]
&=& \Pr \left[ \| f(v)-f(w)\|_q^q - s^qH(t/s) > s^qH(1/s) - s^qH(t/s) + \veps \right]	\\
&<& \Pr \left[ \| f(v)-f(w)\|_q^q - s^qH(t/s) > \veps \right]	\\
&\le& \frac{1}{n^2}.
\end{eqnarray*}

An alternate bound can be derived using Bennett's inequality (Claim \ref{clm:bennett}).
For this, it suffices to take
$k= O \left( \frac{ s^{2q} \log n}{\sigma^2 V \left( \frac{s^q\veps [s^qH(t/s)]}{\sigma^2}\right)} \right)$,
with the variance term $\sigma^2 = \Theta \left( s^{2q} \mathbb{E}[|\sin(ag)|^{2q}] \right)$.
We will prove the case $t \ge 1$, and the case $t<1$ follows as above.
Set $r = \frac{s^q\veps [s^qH(t/s)]}{\sigma^2}$.
If $r \ge 1$, then we have $V(r)=\Omega(r)$, and so
$k
=O \left( \frac{s^q \log n}{\veps [s^qH(t/s)]} \right)
=O \left( \frac{s^q \log n}{\veps} \right)$.
Otherwise $r<1$, and we have $V(r)=\Theta(r^2)$, from which we derive
$k=O \left( \frac{\sigma^2 \log n}{\veps^2 [s^qH(t/s)]^2} \right)$.
Now, if $t \ge s$, we recall by \ref{lem:H}\ref{it:lem-large} that $H(t/s) = \Theta(1)$,
and noting that $\sigma^2 \approx s^{2q} \mathbb{E}[|\sin(ag)|^{2q}] = O(s^{2q})$ we obtain
$k=O \left( \frac{\log n}{\veps^2} \right)$.
When $1 \le t < s$, we have by \ref{lem:H}\ref{it:lem-bilip}\ref{it:lem-bilip2}
that $H(t/s) = \Omega((t/s)^q)$, and so
$k=O \left( \frac{\sigma^2 \log n}{\veps^2 t^{2q}} \right)$.
In this case we require a better bound on $\sigma^2$:
Setting $a=\frac{t}{s}<1$ we have (by analysis similar to the proof of 
Lemma \ref{lem:H})
\begin{eqnarray*}
\sigma^2
&\approx& s^{2q} \mathbb{E}[|\sin(ag)|^{2q}]		\\
&\approx& s^{2q} \left[
a^{2q} \int_0^{1} u^{2q} du
+ a^{2q} \int_1^{1/a} u^{2q-p-1} du
+ \int_{1/a}^{\infty} u^{-p-1} du
\right] \\
&\approx& s^{2q} \left[ \frac{a^{2q}}{2q+1} + \frac{a^p}{2q-p} + \frac{a^p}{p} \right]	\\
&\approx& s^{2q} \frac{a^p}{2q-p}.	
\end{eqnarray*}
It follows that
$k=O \left( \frac{a^{p-2q} \log n}{(2q-p)\veps^2} \right)
=O \left( \frac{s^{2q-p} \log n}{(2q-p)\veps^2} \right)$.
\end{proof}

\section{Range and snowflake embeddings}\label{sec:appext}

In Section \ref{sec:trans} above, we presented our basic embedding.
Here, we show how to use the basic embedding to derive a dimension-reducing
embedding that preserves distances in a fixed range (Section \ref{sec:range}).
We also show that the basic embedding can be used to derive a dimension-reducing
snowflake embedding (Section \ref{sec:snowflake}).
As a precursor to the snowflake embedding, we show that the basic and 
range embeddings can be improved to embed into the doubling dimension of the
space (Section \ref{sec:intdimred}). 

%In this section, we demonstrate that the embeddings of Section \ref{sec:trans} 
%above can be used to embed into the intrinsic dimension, with
%applications to snowflake embeddings (Section \ref{sec:snowflake}) and clustering
%(Section \ref{sec:cluster}).

\subsection{Range embedding}\label{sec:range}
By combining Theorem \ref{thm:full} and Lemma \ref{lem:H}
and choosing an appropriate parameter $s$, we can achieve a
dimension-reducing oblivious strong range embedding for $\ell_p$.
This is the central contribution of our paper:

\begin{theorem}\label{thm:dim-reduction}
Let $1 \le q \le p \le 2$ and $0 < \veps < \frac{1}{2}$,
and consider an $n$-point set $S \in l^m_p$.
Fix range $R > 1$ and set threshold $s$ as follows:
When $q<p$,
$s \approx R \veps^{-1/2} (1+(p-q)\veps^{-(q/2+1)})^{1/(p-q)},$
and when $q=p$,
$s \approx \max \{ R^{1/\veps}, R \veps^{-1/2} e^{\veps^{-(\frac{q}{2}+1)}} \}.$
Then there exists an oblivious embedding $f: S \to l_q^k$ for
\linebreak
$k= O \left( \frac{\log n}{\veps^2} \cdot \min \left\{
s^{2q}, \max \left\{ \frac{s^{2q-p}}{2q-p},\veps s^q \right\}
\right\} \right)$,
which satisfies the following for each point pair $v,w \in S$, where $t=\| v-w \|_p$:
\begin{enumerate}
\renewcommand{\theenumi}{(\alph{enumi})}
\item\label{it:thm-dr-threshold}
Threshold: When $q<p$ we have
$\| f(v)-f(w) \|_q^q \le \frac{s^q}{Q}.$

When $q=p$ we have
$\| f(v)-f(w) \|_q^q \le \frac{s^q}{Q_{R/s}}.$
\item\label{it:thm-dr-contract}
Bounded expansion and contraction for large scales:
When $t>R$, we have
$\|f(v)-f(w)\|_q^q = O(t^q), $ and
$\|f(v)-f(w)\|_q^q \ge R^q + \veps.$
\item\label{it:thm-dr-bilip}
Bi-Lipschitz for intermediate scales:
When $1 \le t \le R$, we have
$(1-\veps) t^q
\le \|f(v)-f(w)\|_q^q
\le (1+\veps) t^q .$
\item\label{it:thm-dr-expansion}
Bounded expansion for small scales: When
$t < 1$, we have
$\|f(v)-f(w)\|_q^q \le 1 + \veps.$
\end{enumerate}
The embedding can be constructed in $O(mk)$ time per point.
\end{theorem}

\begin{proof}
We use the construction of Theorem \ref{thm:full} with the stated value of $s$, and then scale down by a
factor of $Q^{1/q}$ or $Q_{R/s}^{1/q}$. (Note that $Q,Q_{R/s} = \Omega(1)$.) Then the threshold guarantee follows
immediately from Theorem \ref{thm:full}\ref{it:thm-threshold} and the scaling step.

We will now prove the rest of the theorem for the case $p<q$. First, the bi-Lipschitz claim for values
$1 \le t \le R$ follows immediately from Theorem \ref{thm:full}\ref{it:thm-bilip} and
Lemma \ref{lem:H}\ref{it:lem-bilip}, when noting that for an appropriate choice of $s$,
$\frac{t}{s} \le \veps^{1/2} (1+(p-q)\veps^{-(q/2+1)})^{-1/(p-q)}$. Similarly, the bounded expansion
claim for $t<1$ follows from Theorem \ref{thm:full}\ref{it:thm-expansion} and the aforementioned
bi-Lipschitz guarantee at $t=1$. Finally, the bounded expansion for $t>R$ follows from Theorem
\ref{thm:full}\ref{it:thm-bilip}, and the bounded contraction follows from Theorem
\ref{thm:full}\ref{it:thm-bilip} combined with fact that when $a>R/s$, $H(a) > H(R/s)$ (as a consequence
of Lemma \ref{lem:H}\ref{it:lem-bilip2}\ref{it:lem-large} for an appropriate choice of $s$).

For $p=q$, we have essentially the same proof, only noting that the function $Q_a = Q_{t/s}$ is monotone
decreasing in $t$, and since $s = O(R^{1/\veps})$ we have for all $1 \le t \le R$ that
$\frac{Q_{t/s}}{Q_{R/s}}
\le \frac{Q_{1/s}}{Q_{R/s}}
= O \left( \frac{1+\veps^{-1} \log R}{1+ (\veps^{-1} - 1) \log R} \right)
= O(1+\veps)$.
\end{proof}

\subsection{Intrinsic dimensionality reduction}\label{sec:intdimred}

Here we show that the guarantees of Theorem \ref{thm:full} and Theorem \ref{thm:dim-reduction} can be 
achieved by (non-oblivious) embeddings whose target dimension in independent of $n$, and depends only on the 
doubling dimension of the space. 
The following lemma is derived by applying the framework of \cite{GK11} to 
Theorem \ref{thm:full}.

\begin{lemma}\label{lem:fi}
Let $1 \le q \le p \le 2$ and $0 < \veps < \frac{1}{2}$,
and consider an $n$-point set $S \in l^m_p$. Set threshold $s>1$.
Then there exists an embedding $f: S \to l_q^k$ for
$k= \tilde{O} \left( \frac{\ddim^2(S)}{\veps^3} \cdot s \cdot \min \left\{
s'^{2q}, \max \left\{ \frac{s'^{2q-p}}{2q-p},\veps s'^q \right\}
\right\} \right)$ for $s' = \frac{s \ddim (S)}{\veps}$,
which satisfies the following for each point pair $v,w \in S$, where $t=\| v-w \|_p$:
\begin{enumerate}
\renewcommand{\theenumi}{(\alph{enumi})}
\item \label{it:lem-fi-threshold}
Threshold:
$\| f(v)-f(w) \|_q^q < s^q.$
%\item \label{it:lem-fi-bilip}
%Bi-Lipschitz for large scales:
%When $ t \ge s$, we have
%$\|f(v)-f(w)\|_q^q = \Theta(s^q)$
\item \label{it:lem-fi-bilip2}
Bi-Lipschitz for intermediate scales:
When $1 \le t \le s$, we have
$$(1-\veps) s^qH(t/s)
\le \|f(v)-f(w)\|_q^q
\le (1+\veps) s^qH(t/s).$$
\item \label{it:lem-fi-expansion}
Strong bounded expansion for small scales: When
$t < 1$, for some constant $c$
$$\|f(v)-f(w)\|_q^q \le \min \{ (1+\veps)s^qH(1/s), c \ddim(S) t \}.$$
\end{enumerate}
Given a point hierarchy for $S$, the embedding can be constructed in 
$2^{\tilde{O} \ddim(S)} + O(mk)$ time per point.
\end{lemma}

\begin{proof}
Similar to what was done in \cite{GK11}, we compute for $S$ a padded decomposition with padding $s$. 
This is a multiset $[P_1,\ldots,P_m]$ where each partition $P_i$ is a set of clusters, and 
every point is $s$-padded in a $\left( 1- \frac{\veps}{s} \right)$ fraction of the
partitions. Each cluster has diameter bounded by $O \left( s \ddim(S) \right)$, and the support
is $m = \tilde{O}(s \veps^{-1} \ddim(S))$. Using the hierarchy, this can be done in time 
$2^{\tilde{O}(\ddim(S))}$ per point \cite{BGKLR11}.

We embed each partition $P_i$ separately as follows:
For each cluster $C \in P_i$, we extract from $C$ an $\frac{\veps}{\ddim(S)}$-net $N \subset C$
Now each cluster net has aspect ratio $\frac{s \ddim^2 (S)}{\veps}$, and so 
$|N| = \left( \frac{s}{\veps} \right)^{\tilde O\ddim(S)}$. 
We then scale $N$ up by a factor of $\frac{\ddim(S)}{\veps}$ 
so that the minimum distance is 1,
invoke the embedding of Theorem \ref{thm:full} with parameter $s'$ on $N$, 
and scale back down. This procedure has a runtime cost of $O(mk)$ per point,
thresholds all distances at $s$, and reduces dimension to  
$\tilde{O} \left( \frac{\ddim(S)}{\veps^2} \cdot \min \left\{
s'^{2q}, \max \left\{ \frac{s'^{2q-p}}{2q-p},\veps s'^q \right\}
\right\} \right)$. We then concatenate the $m$ cluster partitions together
and scale down by $m^{1/q}$, achieveing an embedding of dimension $k$ for the net points.
We then extend this embedding to all points using the $c\ddim(S)$-factor
Lipschitz extension of Lee and Naor \cite{LN05} for metric space.

For the net points, item \ref{it:lem-fi-threshold} holds immediately. 
For item \ref{it:lem-fi-bilip2}, we note that when 
$1 \le t \le s$, the points fall in the same cluster in a fraction
$\left( 1- \frac{\veps}{s} \right)$ of the partitions, and these 
partitions contribute the correct amount to the interpoint distance.
However, in a fraction $\frac{\veps}{s}$ of the partition the points
are not found in the same cluster, and in these cases the contribution may
be as large as $s$. These partitions account for an additive value of at most 
$\frac{\veps}{s} \cdot s = \veps \le \veps t$. 
Item \ref{it:lem-fi-expansion} follows in an identical manner.

For the non-net points, item \ref{it:lem-fi-threshold} holds since
we may assume that all interpoint distance are at most
$s$, since the non-net points outside the convex hull of the net-points can all
be projected onto the hull, and this can only improve the quality of the extension.
Item \ref{it:lem-fi-bilip2}\ref{it:lem-fi-expansion} follow from 
the embedding of the $\frac{\veps}{\ddim(S)}$-net: By the guarantees of
the extension, an embedded non-net point may be at distance at most $O(\veps)$ from 
its closest net-point, and then the items follow by an appropriate scaling down of $\veps$.
\end{proof}

Similarly, the exact guarantees of Theorem \ref{thm:dim-reduction} can be achieved by a non-oblivious embedding with
dimension independent of $n$.

\begin{corollary}\label{cor:dim-red-intrinsic}
Let $1 \le q \le p \le 2$ and consider an $n$-point set $S \in l^m_p$. Set threshold $s>1$.
Fix range $R>1$ and set threshold $s$ as follows:
When $q<p$,
$s \approx R \veps^{-1/2} (1+(p-q)\veps^{-(q/2+1)})^{1/(p-q)},$
and when $q=p$,
$s \approx \max \{ R^{1/\veps}, R \veps^{-1/2} e^{\veps^{-(\frac{q}{2}+1)}} \}.$
Then there exists an embedding $f: S \to l_q^k$  satisfying items 
\ref{it:thm-dr-threshold}\ref{it:thm-dr-bilip}\ref{it:thm-dr-expansion}
of Theorem \ref{thm:dim-reduction}. The target dimension is
$k= \tilde{O} \left( \frac{\ddim^2(S)}{\veps^3} \cdot s \cdot \min \left\{
s'^{2q}, \max \left\{ \frac{s'^{2q-p}}{2q-p},\veps s'^q \right\}
\right\} \right)$ for $s' = \frac{s \ddim (S)}{\veps}$.
Given a point hierarchy for $S$, the embedding can be constructed in 
$2^{\tilde{O} \ddim(S)} + O(mk)$ time per point.
\end{corollary}

\paragraph{Comment.}
We conjecture that the $\frac{\ddim^2(S)}{\veps^3}$ terms can be reduced to $\frac{\ddim(S)}{\veps^2}$
by combining the randomness used separately for the construction of the padded decomposition,
threshold embedding $f$, and the Lipschitz extension (as was done in \cite{BRS11} for $\ell_2$).

\subsection{Snowflake embedding}\label{sec:snowflake}
The embedding of Lemma \ref{lem:fi}
implies a global dimension-reduction snowflake embedding for $\ell_p$.

\begin{lemma} \label{lem:snowflake}
Let $0<\veps<1/4$, $0<\alpha<1$, and
$\tilde{\alpha} = \min \{ \alpha,1-\alpha\}$.
Every finite subset $S\subset \ell_p$ admits an embedding $\Phi:S\to \ell_q^k$
($1 \le q \le p \le 2$)
for
$k= \tilde{O} \left( \frac{\ddim^6(S)}{\tilde{\alpha}^2 \veps^8} \cdot \min \left\{
s'^{2q}, \max \left\{ \frac{s'^{2q-p}}{2q-p},\veps s'^q \right\}
\right\} \right)$ where $s' = \left( \frac{ \ddim (S)}{\veps} \right)^4$
and $
   1
   \le \frac{\|\Phi(x)-\Phi(y)\|_q}{\|x-y\|_p^{\alpha}}
   \le 1+\veps,
   \qquad \forall x,y\in S.
$
Given a point hierarchy for $S$, the embedding can be constructed in 
$2^{\tilde{O} \ddim(S)} + O(mk)$ time per point.
\end{lemma}

To prove Lemma \ref{lem:snowflake}, we will make use of the snowflake framework presented in
\cite{GK11}. For simplicity, we will first prove the theorem for $\alpha = 1/2$, and then extend
the proof to arbitrary $0<\alpha<1$.
Fix a finite set $S\subset\ell_p$, and assume without loss of generality that the minimum interpoint distance
in $S$ is $1$. Define
$ v = 2 \left\lceil \log_{1+\veps} \left( \frac{\ddim(S)}{\veps} \right) \right\rceil
    = O \left( \tfrac1\veps \log \frac{\ddim(S)}{\veps} \right)$,
and the set
$I=\{i\in\mathbb Z:\ (1+\veps)^{-2v} \le (1+\veps)^i \le (1+\veps)^{2v} \diam(S) \}$.
For each $i\in I$, fix $r_i=(1+\veps)^i$, and let $\varphi_i:S\to\ell_p^{k'}$ 
for $k' = O(k\veps)$ be the embedding that
achieves the bounds of Lemma \ref{lem:fi}
with parameter $s = (1+\veps)^{2v} = \left( \frac{\ddim (S)}{\veps} \right)^4$ and
scaled by $\frac{r_i}{\sqrt{s}}$ so that the bi-Lipschitz property Lemma 
\ref{lem:fi}\ref{it:lem-fi-bilip2} holds whenever
$\frac{r_i}{\sqrt{s}} \le t \le \sqrt{s} r_i$ and the threshold guarantee of
Lemma \ref{lem:fi}\ref{it:lem-fi-threshold} holds at $\sqrt{s}r$.

We shall now use the following technique due to Assouad \cite{A83}.
First, each $\varphi_i$ is scaled down by $1/\sqrt{r_i} = (1+\veps)^{-i/2}$.
They are then grouped in a round-robin fashion into $2v$ groups,
and the embeddings in each group are summed up.
This yields $2v$ embeddings, each into $\ell_p^{k'}$; these
are combined using a direct-sum, resulting in one map $\Phi$ into $\ell_p^{2vk'}$.
Formally, let $i\equiv_{v} j$ denote that two integers $i,j$ are equal
modulo $v$.
Define $\Phi:S\to \ell_p^{2vk'}$ using
the direct sum $\Phi=\bigoplus_{j\in[2v]} \Phi_j$,
where each $\Phi_j:S\to\ell_p^{k'}$ is given by
$$
  \Phi_j=\sum_{i\in I:\ i\equiv_{2v} j} \frac{\varphi_i}{(1+\veps)^{i/2}}.
$$
For $M=M(\veps)>0$ that will be defined later,
our final embedding is $\Phi/\sqrt M:S\to \ell_p^{2vk'}$, which has
target dimension $2vk' = O(k)$, as required (for $\alpha = 1/2$).
It thus remains to prove the distortion bound.
Define $B_i = \frac{\|\varphi_i(x)-\varphi_i(y)\|_p}{(1+\veps)^{i/2}}$.
We will need the following lemma.

\begin{lemma} \label{lem:buckets}
Let $\Phi:S\to\ell_p^{2vk'}$ be as above, let $x,y\in S$.
Then for every interval $A\subset I$ of size $2v$
(namely $A=\{a-v,\ldots,a,\ldots,a+v-1\}$),
\begin{align*}
  \|\Phi(x)-\Phi(y)\|^p_p
  &\le \sum_{i\in A}
        \Big(B_i
        \  + \sum_{i'\in I\setminus A:\ i'\equiv_{2v} i}
        B_{i'}
        \Big)^p \\
  &\le \sum_{i \in A}
        \Big(B_i^p +
        \ 2B_i
        \sum_{i'\in I\setminus A:\ i'\equiv_{2v} i} B_{i'} +
        \Big(
        \sum_{i'\in I\setminus A:\ i'\equiv_{2v} i} B_{i'}
        \Big)^p \Big) \\
  &\le \sum_{a-v \le i < a}
        \Big(B_i^p + B_{i+v}^p +
        \ 2(B_i + B_{i+v})
        \sum_{i'\in I\setminus A:\ i'\equiv_v i} B_{i'} +
        \Big(
        \sum_{i'\in I\setminus A:\ i'\equiv_v i} B_{i'}
        \Big)^p \Big) \\
  \|\Phi(x)-\Phi(y)\|^p_p
  &\ge \sum_{i\in A}
        \Big( \max \Big\{0,
        B_i
        \  - \sum_{i'\in I\setminus A:\ i'\equiv_{2v} i}
        B_{i'}
        \Big\} \Big)^p \\
  &\ge \sum_{a-v \le i < a}
        \Big( B_i^p + B_{i+v}^p -
        \ 2(B_i + B_{i+v})
        \sum_{i'\in I\setminus A:\ i'\equiv_v i}
        B_{i'}
        \Big)
\end{align*}
\end{lemma}
\begin{proof}
By construction,
\[
  \Big\|\Phi(x)-\Phi(y)\Big\|^p_p
  = \sum_{j\in[v]} \Big\|\Phi_j(x)-\Phi_j(y)\Big\|^p_p
  = \sum_{i\in A} \Big\| \sum_{i'\in I:\ i'\equiv_{2v} i}
            \frac{\varphi_{i'}(x)-\varphi_{i'}(y)}{(1+\veps)^{i/2}}
  \Big\|^p_p.
\]
Fix $i\in A$ and let us bound the term corresponding to $i$.
The first required inequality now follows by separating
(among all $i'\in I$ with $i'\equiv_{2v} i$)
the term for $i'=i$ from the rest,
and applying the triangle inequality for
vectors $v_1,\ldots,v_s\in \ell_p^{k'}$, namely,
$
  \|\sum_l v_l\|_p \leq \sum_l \|v_l\|_p
$.
We then bound the sum of the terms for indices $i$ and $i+v$.
The second inequality follows similarly by separating the term for $i'=i$
from the rest, and applying the following triangle inequality for
vectors $u,v_1,\ldots,v_s\in \ell_p^k$, namely,
$
  \|u+\sum_l v_l\|_p \ge \max\{0, \|u\|_p - \sum_l \|v_l\|_p \}
$.
We then bound the sum of the terms for indices $i$ and $i+v$. (Note
that $(\max\{0,b-c\})^p \ge b^p - 2bc$.)
\end{proof}

The proof of Lemma \ref{lem:snowflake} proceeds by demonstrating that, for an
appropriate choice of $A$ (meaning $a$, having fixed $v$), the leading terms in
the above summations ($B_i^p$ and $B_{i+v}^p$ for $a \le i < a+v$) dominate
the sum of all other terms of the summations. Fix $x,y\in S$, let $i^*\in
I$ be such that $(1+\veps)^{i^*} \le \|x-y\|_p \le (1+\veps)^{i^*+1}$, and
let $a=i^*$ so that $A=\{i^*-v,\ldots,i^*+v-1\}$. Consider $i\in A$; we have
the following lemma:

\begin{lemma}\label{lem:sumbounds}
The following hold for all $i^*-v \le i < i^*$:
\begin{enumerate}
\renewcommand{\theenumi}{(\alph{enumi})}
\item\label{it:lowterms}
$\veps \cdot B_i
= \Omega \left( \sum_{i'\in I\setminus A:\ i'\equiv_v i , i'<i} B_{i'} \right) $
\item\label{it:highterms}
$\veps \cdot B_{i+v}
= \Omega \left( \sum_{i'\in I\setminus A:\ i'\equiv_v i , i'>i} B_{i'} \right)$
\item\label{it:allterms}
$\veps \cdot (B_i + B_{i+v})
= \Omega \left( \sum_{i'\in I\setminus A:\ i'\equiv_v i} B_{i'} \right)$
\end{enumerate}
\end{lemma}

\begin{proof}
We prove each item in turn.
\begin{enumerate}
\item[\ref{it:lowterms}]
By Lemma \ref{lem:fi}\ref{it:lem-fi-bilip2} and Lemma \ref{lem:H}\ref{it:lem-bilip2},
we have that $\|\varphi_i(x)-\varphi_i(y)\|_p = \Omega(r_i)$, from which it follows that
$B_i = \Omega  (1+\veps)^{j/2}$.
Similarly we have that for all $j < i$, 
$\|\varphi_j(x)-\varphi_j(y)\|_p = O(\log(\ddim(S)/\veps) \cdot r_j)$, 
from which it follows that $B_j = O( (1+\veps)^{j/2} )$.
Recalling that a geometric series with constant ratio sum to some constant times the largest
largest term, we have that
\[
 \sum_{i'\in I\setminus A:\ i'\equiv_v i , i'<i} B_{i'}
 = O \left( \log(\ddim(S)/\veps) (1+\veps)^{(i-v)/2} \right)
 = O \left( \veps (1+\veps)^{i/2} \right)
 = O \left( \veps B_i \right).
\]

\item[\ref{it:highterms}]
By Lemma \ref{lem:fi}\ref{it:lem-fi-bilip2} and Lemma \ref{lem:H}\ref{it:lem-bilip2},
we have that $\|\varphi_{i+v}(x)-\varphi_{i+v}(y)\|_p = \Omega(\|x-y \|_p) = \Omega((1+\veps)^{i^*})$, so
$B_{i+v} = \Omega \left( (1+\veps)^{i^* - (i+v)/2} \right)$.
By Lemma \ref{lem:fi}\ref{it:lem-fi-expansion}, we have for $j>i$ that
$\|\varphi_{j}(x)-\varphi_{j}(y)\|_p
= O \left( \ddim(S) ( \|x-y \|_p) \right)
= O( \ddim(S) (1+\veps)^{i^*})$. 
It follows that for $i'\equiv_v i$ and $i'>i$ we have 
$B_{i'} = O \left( \ddim(S) (1+\veps)^{i^{*} - i'/2} \right)$, and thus
\[
 \sum_{i'\in I\setminus A:\ i'\equiv_v i , i'>i} B_{i'}
 = O \left( \ddim(S) (1+\veps)^{i^* -(i+2v)/2} \right)
 = O \left( \veps B_{i+v} \right).
\]

\item[\ref{it:allterms}] This follows trivially from parts \ref{it:lowterms} and
\ref{it:highterms}.
\end{enumerate}
\end{proof}

Now, plugging Lemma \ref{lem:sumbounds} into Lemma \ref{lem:buckets}, we obtain
\begin{align*}
  \|\Phi(x)-\Phi(y)\|_p^p
  & \ge \sum_{i\in A:\ i < i^*}
        \Big( B_i^p + B_{i+v}^p - p(B_i+B_{i+v})
        \sum_{i'\in I \setminus A:\ i'\equiv_v i} B_{i'}
        \Big) \\
  &\ge  (1- O(\veps))
        \sum_{i\in A:\ i < i^*}
        \Big( B_i^p + B_{i+v}^p \Big)        \\
  & =   (1- O(\veps))
        \sum_{i\in A} B_i^p    \\
  &=    (1- O(\veps))
	\sum_{i\in A}
	\frac{\|\varphi_i(x)-\varphi_i(y)\|_p}{(1+\veps)^{pi/2}}	\\
  &=	(1- O(\veps)) M \frac{\|x-y\|_p^p}{(1+\veps)^{p{i^*}/2}},
\end{align*}
for some fixed constant $M$, where the last step follows from
Theorem \ref{thm:full} and Lemma \ref{lem:H}\ref{it:lem-smooth}\ref{it:lem-bilip2}.
Similarly, $\|\Phi(x)-\Phi(y)\|_p^p \le (1- O(\veps)) M \frac{\|x-y\|_p^p}{(1+\veps)^{pi/2}}$.
We conclude that the final embedding $\Phi/M^{1/p}$ achieves distortion $1+O(\veps)$
for $\alpha = 1/2$.

\paragraph{Arbitrary $\mathbf{0<\alpha<1}$.}
Turning to proving the theorem for arbitrary values of $0<\alpha<1$,
we repeat the previous construction and
proof with
$v
=2 \frac{\ceil{\log_{1+\veps}(\frac{\ddim(S)}{\veps})}}{\tilde{\alpha}}
=O(\frac{1}{\tilde \alpha \veps} \log \frac{\ddim(S)}{\veps})
$.
%and
%$\delta=(1+\veps)^{-v-1}=\Theta(\veps^{2/ \tilde \alpha})$, where
%$\tilde{\alpha} = \min \{ \alpha,1-\alpha\}$.
As before, $\varphi_i:S\to\ell_2^{k'}$ is the embedding that achieves the
bounds of Theorem \ref{thm:full},
%and $\delta$
%(so $k=\tilde O(\veps^{-3} \log\tfrac1\delta\cdot (\dim S)^2)
%=\tilde O(\veps^{-3} \tilde{\alpha}^{-1} (\dim S)^2)$),
and $\Phi:S\to \ell_p^{vk'}$ is defined by
the direct sum $\Phi=\bigoplus_{j\in[v]} \Phi_j$,
where each $\Phi_j:S\to\ell_p^k$ is given by
$$
  \Phi_j=\sum_{i\in I:\ i\equiv_v j}
\frac{\varphi_i}{(1+\veps)^{i(1-\alpha)}}.
$$
The final embedding is $\Phi/M^{1/p}:S\to \ell_2^{2vk'}$ (for the same $M$
as above), which has the required target dimension.
% $pk \le \tilde O(\veps^{-4} \tilde{\alpha}^{-2}\ddim^2(S))$, as required.

We need to make only small changes to the preceding proof of distortion:
In the statement and proof of Lemma \ref{lem:buckets} and elsewhere, the
dividing term $(1+\veps)^{i/2}$ is replaced by $(1+\veps)^{i(1-\alpha)}$.
Note that the increase in value of $v$ (and the introduction of the term
$\tilde{\alpha}$), is necessary for Lemma \ref{lem:sumbounds} to hold in this
setting. (In particular, the bounds on the geometric series in the proof of
Lemma \ref{lem:buckets} \ref{it:lowterms} require, for the above choice of $v$,
that $\tilde \alpha \le \alpha$, and the bounds
on the geometric series in the proof of \ref{it:highterms} require
$\tilde \alpha \le 1-\alpha$.) No other changes to the proof are necessary, and this
completes the proof of Lemma \ref{lem:snowflake}.

\paragraph{Comments.}
For $l_1$, the snowflake dimension can be further improved by plugging in the dimension reduction embedding of
\cite{OR02}. For $\alpha \le \frac{p}{2}$, since the $\alpha$-snowflake of $\ell_p$ is in $l_2$ \cite{DL97},
one could instead embed the snowflake into $l_2$, reduce dimension via \cite{JL84} or via a second snowflake from $l_2$ to
$l_2$ \cite{BRS11,GK11}, and embed back into $\ell_p$ using \cite{JS82}.
In regards to Lemma \ref{lem:fi} above, we conjecture that the target dimension in Lemma \ref{lem:snowflake} can be reduced to
$k=\tilde O(\veps^{-3}\tilde{\alpha}^{-2}(\dimC S))$ by combining randomness.
%Such a result would yield an improvement over
%the ordinal embedding for $\ell_p$ presented in \cite{BDHSZ08}.

\section{Clustering}\label{sec:cluster}
Here we show that our snowflake embedding can be used to 
produce faster algorithms for the $k$-center and min-sum clustering problems.
In both cases, we obtain improvements whenever
$(\ddim/\veps)^{\Theta(1)}$ is smaller than the ambient dimension.

\paragraph{$k$-center clustering.}
In the $k$-center clustering problem, the goal is to partition the input set into $k$ clusters,
where the objective function to be minimized is the maximum radius among the clusters.
Agarwal and Procopiuc \cite{AP02} considered this problem for 
$d$-dimensional set $S \subset \ell_p$, and gave an exact algorithm that runs in time
$n^{O(k^{1-1/d})}$,
and used this to derive a $(1+\veps)$-approximation algorithm that runs in 
$O(nd\log k) + (k\veps^{-d})^{O(k^{1-1/d})}$. 
Here, the cluster centers are points of the ambient space $\R^d$ in which $S$ resides, 
chosen to minimize the maximum distance from points in $S$ to their nearest center.
The authors claim that the algorithm can be applied to the {\em discrete} problem as well, where all
centers are chosen from $S$, and in fact the algorithm applies to the more
general problem where the centers are chosen from a set $S'$ satisfying $S \subset S' \subset \R^d$.%
\footnote{The Euclidean core-set algorithm of \cite{BHI02} runs in time $k^{O(k/\veps^2)} \cdot nd$, and 
can readily be seen to apply to all $\ell_p$ for constant $p$, $1<p \le 2$ 
(see \cite{Pa04} for a simple approach). The algorithm of \cite{AP02} compares favorably to the core-set 
algorithm when $d$ is small.}
Clearly, the runtime of the algorithm can be improved if the dimension is lowered.

\begin{theorem}
Given $d$-dimensional point set $S \subset \ell_p$ for constant $p$, $1<p \le 2$, a 
$(1+\veps)$-approximate solution to the $k$-center problem on set $S$ can be computed in time
$O(nd (2^{\tilde{O}(\ddim(S))} + \log k)) 
+ (k \cdot \veps^{-\ddim(S)\log(1/\veps)/\veps^2})^{k^{1-(\veps/\ddim(S))^{O(1)}}}$.
This holds for the discrete case as well.
\end{theorem}

\begin{proof}
We first consider the discrete case. Let $r^*$ be the optimal radius.
As in \cite{AP02}, we run the algorithm of Feder and Greene \cite{FG88} in 
time $O(n\log k)$ to obtain a value $\tilde{r}$ satisfying
$r^* \le \tilde{r} < 2r^*$. 
We then build a hierarchy and extract a $\frac{\veps}{2} \tilde{r}$-net $V \subset S$ in time 
$2^{\tilde{O}(\ddim(S))}nd$ (assuming the word-RAM model \cite{BG13}). 
Since all points of $S$ are contained in $k$ balls of radius $r$, we have
$|V| = k\veps^{-O(\ddim(S))}$. Further, a $k$-clustering for $V$ is a
$(1+O(\veps))$-approximate $k$-clustering for $S$. 
We then apply the snowflake embedding of Lemma \ref{lem:snowflake} to embed $V$ into 
$(\ddim(S)/\veps)^{O(1)}$-dimensional $\ell_p$,
and run the exact algorithm of \cite{AP02} on $V$ in the embedded space. 
Since a snowflake perserves the ordering of distances, 
the returned solution for the embedded space 
is a valid solution in the origin space as well.

Turning to the general (non-discrete) case, the above approach is problematic in that the
embedding makes no guarantees on embedded points not in $V$. To address this, we construct a set 
$W$ of candidate center points in the ambient space $\R^d$:
Recall that the problem of finding the minimum enclosing $\ell_p$-ball admits a core-set of 
size $O(\veps^{-2})$ \cite{BHI02}. (That is, for any discrete point set there exists a subset 
of size $O(\veps^{-2})$ with the property that the center of the subset is also the center of a 
$(1+\veps)$-approximation to the smallest ball covering the original set.)
We take all distinct subsets $C \subset V$ of size $|C| = O(\veps^{-2})$ and radius at most
$\tilde{r}$, compute the center point of each subset (see \cite{Pa04}), and add its candidate 
center to $W$. It follows that 
$|W| = k\veps^{-O(\ddim(S)/\veps^2)}$
and 
$\ddim(W) = \log \veps^{-O(\ddim(S)/\veps^2)} = O(\ddim(S) \log(1/\veps) / \veps^2)$.
As above, we use the snowflake embedding of Lemma \ref{lem:snowflake} to embed $V \cup W$ into
$(\ddim(S)/\veps)^{O(1)}$-dimensional $\ell_p$ 
and run the exact algorithm of \cite{AP02} on the embedded space, covering the 
points of $V$ with candidate centers from $W$.
The returned solution is a valid solution in the origin space as well.
\end{proof}

\paragraph{Min-sum clustering.}
In the min-sum clustering problem, the goal is to partition the points of an
input set into $k$ clusters, where the objective function is the sum of 
distances between each intra-cluster point pair.
Schulman \cite{Schulman00} designed algorithms
for min-sum clustering under $\ell_1,\ell_2$ and $\ell_2$-squared costs,
and their runtimes depend exponentially on the dimension. 
We will to obtain faster runtimes for min-sum clustering for $\ell_0$ 
($1 \le p \le 2$) by
using our snowflake embedding as a preprocessing step to reduce dimension.
Ultimately, we will embed the space into $\ell_2$ using a $\frac{1}{2}$-snowflake, and
then solve min-sum clustering with $\ell_2$-squared costs in the embedded 
space; this is equivalent to solving the original $\ell_p$ problem.
We will prove the case of $\ell_1$, and the other cases are simpler.

We are given an input set $S \in \ell_1$, and set 
$c = 2-\frac{1}{d'}$ for some value $d' = \left( \frac{\ddim(S)}{\veps} \right)^{O(1)}$.
We note that the $\frac{1}{c}$-snowflake of $\ell_1$ is
itself in $\ell_c$ \cite{DL97}, and also that this embedding into $\ell_c$ 
can be computed in polynomial time (with arbitrarily small distortion) 
using semi-definite programming \cite{GK11}, although the target dimension may be large. 
We compute this embedding, and then reduce dimension by
invoking our snowflake embedding (Lemma \ref{lem:snowflake}) to compute a 
($\frac{c}{2} = 1 - \frac{1}{2d'}$)-snowflake in
$\ell_c$, with dimension $d'$.
We then consider the vectors to be in $\ell_2$ instead of $\ell_c$,
which induces a distortion of $1+O(\veps)$.
Finally, we run the algorithms of Schulman on the final Euclidean space.
As we have replaced the original $\ell_1$ distances with their
$\left( \frac{1}{c} \cdot \frac{c}{2} = \frac{1}{2} \right)$-snowflake and 
embedded into $\ell_2$, solving min-sum clustering with $\ell_2$-squared costs
on the embedded space solves the original $\ell_1$ problem with distortion
$1+O(\veps)$.

The following lemma follows from the embedding detailed above, 
in conjunction with \cite[Propositions 14,28, full version]{Schulman00}. 
For ease of presentation, we will assume that $k=O(1)$.

\begin{lemma}
Given a set of $n$ points $S \in \R^d$, a $(1+O(\veps))$-approximation to
the $\ell_p$ min-sum $k$-clustering for $S$, for $k=O(1)$, can be computed
\begin{itemize}
\item
in deterministic time
$n^{O(d')} 2^{2^{(O(d'))}} $.
\item
in randomized time
$n^{O(1)} + 2^{2^{(O(d'))}} (\veps \log n / \delta)^{O(d')}$,
with probability $1-\delta$.
\end{itemize}
where $d' = (\ddim(S)/\veps)^{O(1)}$.
\end{lemma}

\paragraph{Acknowledgements.} We thank Piotr Indyk, Robi Krauthgamer
and Assaf Naor for helpful conversations.

{\small %\footnotesize
\bibliographystyle{alpha}
\bibliography{bib-lp}
}

\appendix

\section{Probability theory.}

The following is a simplified version of Hoeffding's inequality \cite{Ho63, Lu04}:

\begin{claim}\label{clm:hoeffding}
Let $X_1,\ldots,X_k$ be independent real-valued random variables, and
assume $|X_i| \le s$ with probability one. Let $\bar{X} = \frac{1}{k}\sum_{i=1}^k X_i$.
Then for any $z>0$,
\begin{equation}
\Pr \left[| \bar{X} - \mathbb{E}[\bar{X}] | \ge z \right]
\le 2\exp \left( - \frac{2 k z^2}{s^2} \right)
\end{equation}
\end{claim}

The following is a restatement of Bennett's inequality \cite{Be62, Lu04}:

\begin{claim}\label{clm:bennett}
Let $X_1,\ldots,X_k$ be independent real-valued random variables, and
assume $|X_i| \le s$ with probability one. Let
$\bar{X} = \frac{1}{k}\sum_{i=1}^k X_i$ and set
$\sigma^2 = \frac{1}{k} \sum_{i=1}^k \var \{ X_i \}$. Then for any $z>0$,
\begin{equation}
\Pr \left[| \bar{X} - \mathbb{E}[\bar{X}] | \ge z \right]
\le 2 \exp \left( - \frac{k \sigma^2}{s^2} V\left( \frac{sz}{\sigma^2} \right) \right)
\end{equation}
where $V(u) = (1+u) \ln (1+u) -u$ for $u \ge 0$.
\end{claim}

Note that for $u \ge 1$, we have $V(u) = \Omega(u \log u)$, while
for $0\le u<1$, $V(u) = \Theta(u^2)$.

\end{document}